\documentclass{article}

\usepackage{microtype}
\usepackage{graphicx}
\usepackage{subfigure}
\usepackage{booktabs} 

\usepackage{hyperref}

\usepackage[accepted]{icml2018}
\usepackage{amsmath,amsthm,amssymb}
\usepackage[titlenumbered,ruled]{algorithm2e}

\newcommand{\R}{\mathbb{R}}

\makeatletter
\def\thm@space@setup{\thm@preskip=0pt
\thm@postskip=10pt}
\makeatother
\newtheorem{theorem}{Theorem}
\newtheorem{lemma}[theorem]{Lemma}

\newtheorem{claim}[theorem]{Claim}

\theoremstyle{definition}
\newtheorem{definition}{Definition}

\icmltitlerunning{Utility Preserving Secure Private Data Release}

\begin{document}

\twocolumn[
\icmltitle{Utility Preserving Secure Private Data Release}



\icmlsetsymbol{equal}{*}

\begin{icmlauthorlist}
\icmlauthor{\textbf{Jasjeet Dhaliwal}}{to}
\icmlauthor{\textbf{Geoffrey So}}{to}
\icmlauthor{\textbf{Aleatha Parker-Wood}}{to}
\icmlauthor{\textbf{Melanie Beck}}{to}
\end{icmlauthorlist}

\icmlaffiliation{to}{Center for Advanced Machine Learning, Symantec Corporation, Mountain View, California, USA}

\icmlcorrespondingauthor{Jasjeet Dhaliwal}{jasjeet\_dhaliwal@symantec.com}

\icmlkeywords{Machine Learning, ICML}

\vskip 0.3in
]



\printAffiliationsAndNotice{}  

\begin{abstract}
Differential privacy mechanisms that also make reconstruction of the data impossible come at a cost - a decrease in utility. In this paper, we tackle this problem by designing a private data release mechanism that makes reconstruction of the original data impossible and also preserves utility for a wide range of machine learning algorithms. We do so by combining the Johnson-Lindenstrauss (JL) transform with noise generated from a Laplace distribution.  While the JL transform can itself provide privacy guarantees \cite{blocki2012johnson} and make reconstruction impossible, we do not rely on its differential privacy properties and only utilize its ability to make reconstruction impossible. We present novel proofs to show that our mechanism is differentially private under single element changes as well as single row changes to any database.  In order to show utility, we prove that our mechanism maintains pairwise distances between points in expectation and also show that its variance is proportional to the dimensionality of the subspace we project the data into. Finally, we experimentally show the utility of our mechanism by deploying it on the task of clustering. 
\end{abstract}

\section{Introduction}
\label{introduction}
While the recent surge in data available for machine learning has created new opportunities, it also poses a risk to the privacy of individuals. Differential privacy \cite{dwork2006calibrating} is the most widely accepted framework that attempts to address this issue by capturing precisely how much additional information of an individual is leaked by participating in a database that would not have been leaked otherwise. Consequently, by providing guarantees within the framework of differential privacy, one can protect the privacy of the individuals participating in a database while utilizing the private data to build machine learning models. 

However, differentially private releases of aggregate statistics have been shown to be susceptible to reconstruction and tracing attacks \cite{dwork2017exposed} under certain constraints. This is also true in the non-interactive setting where the amount of noise added to the data must be increased drastically in order to make reconstruction difficult \cite{bhowmick2018protection, dwork2014algorithmic}. Even differentially private machine learning models have been shown to be susceptible to membership inference attacks that leak information about individual participation \cite{rahman2018membership, shokri2017membership}. The above issues create the need for stronger privacy mechanisms that make it provably impossible for an adversary to reconstruct the original data, hence eliminating concerns about privacy. However, mechanisms that do so, achieve this goal at the cost of a significant drop in utility \cite{blocki2012johnson}. One approach to improve this utility has been to design private data release mechanisms that are tailored to specific machine learning algorithms \cite{upadhyay2014differentially, blocki2012johnson,sheffet2015private}. But this approach limits the ability of the analyst to compare different machine learning algorithms on a given dataset. For instance, if an analyst wants to perform clustering in order to better understand the data before classification,  the covariance matrix or other similar aggregate statistics may not suffice. Similarly, an analyst may wish to compare a neural network, linear regression, and a random forest on a given dataset before deploying a model. It is therefore important to design a mechanism that privately releases data in a form that it can be used by a wide variety of machine learning algorithms while preserving utility.

 In this paper, we tackle this problem of preserving utility for a broad class of machine learning algorithms while also making reconstruction of the original data impossible. We do so by using the JL transform in conjunction with the Laplace mechanism. The JL transform is a powerful dimensionality reduction tool that preserves pairwise distances between points \cite{dasgupta2003elementary} and has also been shown to provide differential privacy guarantees by itself \cite{blocki2012johnson}. We do not rely on the differential privacy guarantees of the JL transform but utilize the fact that it makes it impossible to reconstruct the exact original values (or the original dimensionality) of the data 
\cite{liu2006random}. We do not use the differential privacy properties of the JL transform because in order to achieve differential privacy via the JL transform by itself, certain constraints must be imposed on the the rank and the spectrum of the data matrix. More specifically, the data matrix is required to be full rank and the smallest eigenvalue of the data matrix must be larger than a given threshold.  For data matrices that do not meet the required constraints, the mechanism in \cite{blocki2012johnson} modifies the spectrum of the matrix via an operation that greatly compromises utility. This leaves open the problem of utilizing the ability of the JL transform to make reconstruction impossible while providing differential privacy via a mechanism that preserves utility.

We make progress in this direction by combining the power of the JL transform and the Laplace mechanism in a manner that preserves utility while also providing differential privacy guarantees such that the original data cannot be reconstructed. Our approach is similar to \cite{kenthapadi2012privacy} and can be considered as an extension to their work. Our data release mechanism provides utility for general machine learning tasks and is differentially private under single element changes as well as row changes. We prove the differential privacy guarantees provided by our mechanism and also propose an algorithm that maintains pairwise distances between private data points in expectation. We chose the most general task of clustering in order to show the effectiveness of our methods experimentally.

Our contributions in this paper are: 
\begin{itemize}
    \item We propose a differentially private data release mechanism that preserves privacy and makes it impossible for an adversary to reconstruct the original data.
    \item We propose a distance recovery algorithm and prove that it maintains pairwise distances in expectation. We also prove precise variance guarantees for the distance recovery algorithm.
    \item We experimentally validate the utility of our mechanism by showing that it maintains pairwise distances and performs well on the general task of clustering
\end{itemize}

The paper is organized as follows: we first provide the required background on differential privacy in Section 2. We describe our mechanism in Section 3 and prove its differential privacy guarantees. We provide  utility guarantees in Section 4 and experimentally validate our claims in Section 5, showing the effectiveness of our mechanism. Related work is covered in Section 6, followed by the Conclusion in Section 7. 

\section{Background}
In this section we define differential privacy and also cover other necessary mathematical background required for our results.

\subsection{Differential Privacy}
Differential privacy captures precisely how likely is it for a third-party to ascertain whether an individual participated in a database or not. In order to formalize the definition of differential privacy, we first introduce the notion of neighboring databases. \\

\theoremstyle{definition}
\begin{definition}[Neighboring databases]
Given an input space $\mathcal{X} \subseteq \mathbb{R}^{d}$, we can represent a database with $n$ entries,  $X \in \mathcal{X}^{n}$ as $X \in \mathbb{R}^{n \times d}$. Then, two databases $X_1, X_2 \in \mathbb{R}^{n \times d}$ are row-wise neighbors if they differ in exactly one row. They are considered element-wise neighbors, if they differ in exactly one element.
\end{definition}

\theoremstyle{definition}
\begin{definition}[Probability Simplex]
Given a set $\mathcal{Y}$, the probability simplex over $\mathcal{Y}$ is defined as : $\Delta{\mathcal{Y}} = \left\{ \mathbf{y} \in \mathbb{R}^{|\mathcal{Y}|}: y_i \geq 0, \sum_{i=1}^{|\mathcal{Y}|}y_i = 1  \right\}$
\end{definition}

\theoremstyle{definition}
\begin{definition}[Randomization Mechanism]
Given two sets $\mathcal{X}, \mathcal{Y}$, a randomization mechanism is a function $\mathcal{M}: \mathcal{X} \to \Delta{\mathcal{Y}}$.
\end{definition}

Thus, a randomization mechanism defines a probability distribution over the set $\mathcal{Y}$. Given an input $ x \in \mathcal{X}$, a randomization mechanism $\mathcal{M}$, maps $x$ to $y \in \mathcal{Y}$ with probability $(\mathcal{M}(x))_{y}$, which is the probability for element $y$ under the distribution $(\mathcal{M}(x))$.\\
\theoremstyle{definition}
\begin{definition}[Privacy Loss]
For a randomization mechanism $\mathcal{M}$, the privacy loss for two neighboring databases $X_1, X_2$, is defined as: $\mathcal{L}^{D}_{\mathcal{M}(X_1) || \mathcal{M}(X_2)} = \text{ln} \left (\frac{\mathbf{P}[\mathcal{M}(X_1) \subseteq D] -\delta}{\mathbf{P}[\mathcal{M}(X_2) \subseteq D]}\right) $
\end{definition}

\theoremstyle{definition}
\begin{definition}[Differential Privacy]
For any $\epsilon > 0$, and $\delta \in [0,1)$, a randomization mechanism $\mathcal{M}$ is $(\epsilon, \delta)$ differentially private on domain $\mathcal{X}$ if for two neighboring databases $X_1, X_2$, the privacy loss  $\left |\mathcal{L}^{D}_{\mathcal{M}(X_1) || \mathcal{M}(X_2)} \right | \leq \epsilon$.
\end{definition}

 For a more thorough review of $(\epsilon, \delta)$ privacy, the reader is referred to  \cite{dwork2014algorithmic}. 
\subsection{Johnson-Lindenstrauss Lemma}
The Johnson-Lindenstrauss Lemma states that a set of points in a high-dimensional space can be embedded into a lower dimensional space such that the distances between the projected points are preserved with high probability. We provide a statement of this lemma that relies on a projection matrix using values from the Gaussian distribution.

Consider a finite set $S \subset \mathbb{R}^{d}$ with $|S| = n$. Let $P \in \mathbb{R}^{d \times k}$ be a real valued matrix such that $P_{ij} \sim \mathcal{N}(0, \frac{1}{k})$, where $k = \Omega(\Lambda^{-2} \text{log}(n))$ for $0 < \Lambda \leq 1$. Then for any $x,y \in S$, we have:

$(1 - \Lambda)||x-y||_{2}^{2} \leq ||xP - yP||_2^2 \leq (1 + \Lambda)||x-y||_{2}^{2}$

Further, $\mathbb{E}\left[||xP - yP||_2^{2}\right] = ||x - y||_2^{2}$.  This result is called the Johnson-Lindenstrauss (JL) lemma. We refer the reader to \cite{dasgupta2003elementary} for a proof of the lemma.

\section{Privacy Guarantees}
Our randomization mechanism  utilizes the JL transform to reduce the dimensionality of the input and then uses the Laplacian mechanism to provide differential privacy guarantees. Since the elements of the JL matrix are normally distributed, we utilize their properties in conjunction with the Laplace mechanism to provide differential privacy guarantees while still maintaining utility. We do not rely on the differential privacy properties of the JL transform itself \cite{blocki2012johnson} because it requires a transformation of the data matrix that does not preserve any utility in certain practical cases (see Section \ref{related_work} for details).

Instead our mechanism design follows that of \cite{kenthapadi2012privacy} with two key differences: a) our mechanism adds noise from the Laplace distribution (whereas \cite{kenthapadi2012privacy} added noise from a Gaussian distribution) b) our mechanism provides privacy guarantees with respect to element and row-wise changes (whereas \cite{kenthapadi2012privacy} only provide guarantees with respect to element wise changes). The JL transform not only reduces dimensionality of the input, but also provides further security from attackers by making it impossible to reconstruct the original data values if the JL transformation matrix is kept secret \cite{liu2006random}. We now describe our randomization mechanism.

\subsection{Randomization Mechanism}
Given database $X \in \R^{n \times d}$, our mechanism first projects the data onto a lower dimensional subspace $\R^{k}$, with $k << d$, and then adds a noise matrix $\Delta \in \R^{n \times k}$ to the projected data. The entries of this noise matrix are drawn i.i.d from a Laplacian distribution. The mechanism requires the projection parameter $k$ which determines the dimensionality of the subspace that we wish to project the data into. In addition, it requires the privacy parameters $c$ and $\epsilon$ in order to determine the scale of the Laplacian distribution, where $\epsilon$ is determined by the level of privacy we wish to maintain and $c$ is a parameter that will become clear in the proofs of privacy guarantees. Algorithm \ref{alg:algorithm1}. outlines our mechanism.  \\

\begin{algorithm}
	\SetKwInOut{Input}{Input}
	\SetKwInOut{Output}{Output}

	\Input{$X \in \R^{n \times d}, k,c, \epsilon$}
	\Output{ $Z \in \R^{n \times k}$}
    
    \begin{enumerate}
    \item Construct JL projection matrix $P \in \R^{d \times k}$ such that $P_{ij} \sim \mathcal{N}(0, \frac{1}{k})$
    \item Set $Y = XP$
    \item Construct noise matrix $\Delta \in \R^{n \times k}$ such that $\Delta_{ij} \sim \text{Laplacian}(0,\frac{c}{\epsilon})$.
    \item Return $Z = Y + \Delta $
    \end{enumerate}
    
 \caption{Randomization Mechanism}
 \label{alg:algorithm1}
\end{algorithm} 


Note that we do not release the projection matrix $P$, in order to eliminate the possibility of a reconstruction attack \cite{dwork2008new}.

\subsection{Privacy Guarantees}
\begin{lemma}\label{lemma1}
 For any $X, X' \in \R^{n \times d}$, such that $X$ and $X'$ differ in exactly one element with $||X - X'||_1 \leq 1$, we have for any $A \in \R^{d \times k}$,  $||XA - X'A||_1 \leq \sqrt{k} \max\limits_{1 \leq i \leq d} ||A_i||_2$, where $A_i$ is the $i$th row of $A$.
\end{lemma}

\begin{proof}
We prove the above by direct calculation. 
\begin{align*}
||XA - X'A||_1  &= ||(X - X')A||_1 \nonumber \\
&\leq \max\limits_{1 \leq i \leq d} \sum_{j=1}^{k} |A_{ij}| \nonumber \\
&= \max\limits_{1 \leq i \leq d} ||A_i||_1 \nonumber \\
&\leq \sqrt{k} \max\limits_{1 \leq i \leq d} ||A_i||_2 
\end{align*}

where the second inequality follows from the fact that $(X - X')$ only contains one non-zero element and the last inequality follows from the Cauchy-Schwarz inequality for inner product spaces.
\end{proof}

\begin{lemma}\label{lemma2} \cite{kenthapadi2012privacy} Let $P \in \R^{d \times k}$ such that $P_{ij} \sim \mathcal{N}(0, \frac{1}{k})$, then $\mathbf{Pr}\left[\max\limits_{1 \leq i \leq d} ||P_i||_2 > 
1 + \sqrt{\frac{2x}{k}}\right] < de^{- x}$, for any $x > 0$.
\end{lemma}

\begin{theorem}\label{theorem3}
For any two element-wise neighboring databases $X,X' \in \R^{n \times d}$, such that $||X - X'||_1 \leq 1$, Algorithm 1. achieves $\epsilon$-differential privacy with probability at least $1 -  de^{-\frac{k}{2}}$,  with respect to changes in a single element.
\end{theorem}
\begin{proof}
 Using Lemma \ref{lemma1} we have, $||XP - X'P||_1  \leq 
\sqrt{k} \max\limits_{1 \leq i \leq d} ||P_i||_2 $. Let, $Y = XP$ and $Y' = X'P$ and without loss of generality, consider $Y,Y', \Delta \in \R^{nk}$. Now, set $Z = Y + \Delta$,  $Z' = Y' + \Delta$, and let $D \subset \R^{nk}$. Due to the i.i.d assumption on the elements of $\Delta$ we have:

\begin{align*}
\mathbf{Pr}[Z \in D] 
&= \frac{1}{(2b)^{nk}}\int_{D} e^{-\frac{1}{b}(||z - Y||_1)}dz \nonumber \\
&\geq \frac{1}{(2b)^{nk}}\int_{D} e^{-\frac{1}{b}(||z - Y'||_1 + ||Y' - Y||_1)}dz \nonumber \\
&= \frac{1}{(2b)^{nk}}\int_{D} e^{-\frac{1}{b}(||z - Y'||_1)} e^{-\frac{1}{b}(||Y' - Y||_1)}dz \nonumber \\
&= e^{-\frac{1}{b}(||Y' - Y||_1)}P[Z' \in D] \nonumber \\
\end{align*}

$\implies \mathbf{Pr}[Z' \in D] \leq e^{\frac{1}{b}(||Y' - Y||_1)}\mathbf{Pr}[Z \in D]$

We set $\epsilon = \frac{c}{b}$ in Algorithm \ref{alg:algorithm1}, therefore, in order to preserve privacy, we must constrain $\frac{||Y - Y'||_1}{b} < \frac{c}{b}$. 

\begin{align}
\mathbf{Pr}\left[\frac{||Y - Y'||_1}{b} > \frac{c}{b}\right] &=\mathbf{Pr}[||Y - Y'||_1 > c] \nonumber \\
&\leq \mathbf{Pr}\left[  \max\limits_{1 \leq i \leq d} ||P_i||_2 > \frac{c}{\sqrt{k}}\right] \nonumber 
\end{align}

Setting $c = 2\sqrt{k}$, and $x$ in Lemma \ref{lemma2} to $\frac{k}{2}$, we get: 

\begin{align}
  \mathbf{Pr}\left[  \max\limits_{1 \leq i \leq d} ||P_i||_2 > 2\right] \leq de^{-\frac{k}{2}}\nonumber
\end{align}

Hence, Algorithm 1. achieves $\epsilon$-differential privacy with probability at least $1 -  de^{-\frac{k}{2}}$. 

\end{proof}

Thus, Theorem \ref{theorem3} provides privacy guarantees and also utilizes the dimensionality reduction properties of the JL transform. It is similar to the work done by \cite{kenthapadi2012privacy} in which the authors define a random mechanism that first does a JL transform and then adds Gaussian noise. Since the above result is limited to providing privacy guarantees for element-wise changes, we now focus on extending it to row-wise changes. That is, we now show that Algorithm 1. provides differential privacy when two databases differ by one row. \\

\begin{lemma}\label{lemma4}
 If $P \in \R^{d \times k}$ with $P_{ij} \sim \mathcal{N}(0, \frac{1}{k})$, and $v \in R^d$ then  $\mathbf{Pr} \left[k \max\limits_{1 \leq i \leq k} \left| \sum_{j=1}^{d} v_jP_{ji} \right| > kt \right] \leq 2ke^{\frac{-kt^2}{2||v||^2_2}}$. 
\end{lemma}
\begin{proof}
First note that, $ \sum_{j=1}^{d} v_jP_{ji} \sim \mathcal{N}(0,\frac{||v||^2_2}{k})$. Therefore, it is Gaussian and hence Sub-Gaussian, which lets us use tail bounds for Sub-Gaussian random variables and get: 

$\mathbf{Pr} \left[\left| \sum_{j=1}^{d} v_jP_{ji} \right|> t \right] \leq 2e^{\frac{-kt^2}{2||v||^2_2}}$. Using the union bound and multiplying both sides by $k$, we get the desired result. \\
\end{proof}

We now show that our mechanism provides differential privacy guarantees with respect to row changes in the data matrix. 

\begin{theorem}\label{theorem5}
For any two row-wise neighboring databases $X,X' \in \R^{n \times d}$,  that differ in row $m$, such that $||X_m - X_m'||^2_2 \leq \alpha$,  Algorithm 1. achieves $\epsilon$-differential privacy with probability at least $1 - 2ke^{\frac{-kt^2}{2\alpha}}$, where $t \geq \sqrt{\frac{2 \text{ln }2k}{k}\alpha}$.
\end{theorem}

\begin{proof}
Suppose that $X, X'$ differ in row $m$, then we have: 
\begin{align}
||XP - X'P||_1  &= \sum_{i=1}^{k}\left| \sum_{j=1}^{d} P_{ji} (X_{mj} - X_{mj}') \right| \nonumber \\
&\leq k \max\limits_{1 \leq i \leq k} \left| \sum_{j=1}^{d} P_{ji} (X_{mj} - X_{mj}') \right|    \nonumber
\end{align}

Next, we can see that $ \sum_{j=1}^{d} P_{ji} (X_{mj} - X_{mj}') \sim \mathcal{N}(0,\frac{||X_m - X_m'||^2_2}{k})$. Hence, we can use Lemma \ref{lemma4} to get:  

\begin{align}
&\mathbf{Pr} \left[k \max\limits_{1 \leq i \leq k}\left| \sum_{j=1}^{d} P_{ji} (X_{mj} - X_{mj}')\right|> kt \right] \nonumber \\ 
&\leq 2ke^{\frac{-kt^2}{2||X_m - X_m'||^2_2}} \nonumber
\end{align} 

Setting $c = kt$, and $t \geq \sqrt{\frac{2 \text{ln }2k}{k}\alpha}$ we ensure that $2ke^{\frac{-kt^2}{2||X_m - X_m'||^2_2}} \in [0,1]$. Once again letting $Y = XP$ and $Y' = X'P$ and without loss of generality, letting $Y,Y', \Delta \in \R^{nk}$, we set $Z = Y + \Delta$,  $Z' = Y' + \Delta$, and let $D \subset \R^{nk}$. By following the same steps as we did in Theorem \ref{theorem3}. we get:
\begin{align}
\mathbf{Pr}\left[Z' \in D\right] &\leq e^{\frac{1}{b}(||Y' - Y||_1)}\mathbf{Pr}\left[Z \in D\right]  \nonumber
\end{align}

Now, we constrain $\frac{||Y - Y'||_1}{b} < \frac{c}{b}$.

\begin{align}
&\mathbf{Pr}\left[\frac{||Y - Y'||_1}{b} >\frac{c}{b}\right] \nonumber \\
&=\mathbf{Pr}\left[||Y - Y'||_1 > c\right] \nonumber \\
&=\mathbf{Pr}\left[||(X - X')P||_1 > c \right] \nonumber \\
&\leq \mathbf{Pr}\left[ k \max\limits_{1 \leq i \leq k}\left| \sum_{j=1}^{d} P_{ji} (X_{mj} - X_{mj}') \right| > c \right] \nonumber \\
&\leq 2ke^{\frac{-kt^2}{2||X_m - X_m'||^2_2}} \nonumber \\
&\leq 2ke^{\frac{-kt^2}{2\alpha}} \nonumber
\end{align}

Hence, Algorithm 1. achieves $\epsilon$-differential privacy with probability at least $1 - 2ke^{\frac{-kt^2}{2\alpha}}$.

\end{proof}

\section{Utility Guarantees}
A differentially private mechanism that is also an isometric isomorphism would allow any machine learning algorithm to extract the same amount of utility from the private data as it could from the non-private data. Drawing from that intuition, we also measure utility as was proposed in \cite{kenthapadi2012privacy}, by the degree to which a privacy mechanism preserves pairwise distances after its action. That is, a mechanism that allows pairwise distances to be preserved in the private representation of the data is more useful than one that does not. In order to capture this notion, we first define a distance recovery algorithm in Algorithm \ref{alg:algorithm2} that takes as input, two private data points and outputs the distance between them. We then show that the algorithm preserves squared distances in expectation. Further, we show that the variance of the squared distance between any two points is proportional to the dimensionality of the subspace that the mechanism projects the data into. 

\begin{algorithm}\label{algorithm2}
	\SetKwInOut{Input}{Input}
    \SetKwInOut{Output}{Output}

	\Input{$Z \in \R^{n \times k}$, $\sigma^{2}$, $(i,j) \in \{1, \dots, n\} \times \{1, \dots, n\}$ }
	\Output{Distance between $Z_i$ and $Z_j$}
  
    \begin{enumerate}
    \item Output $\mathcal{D}(Z_i,Z_j) = ||Z_i - Z_j||_{2}^{2} - 2k\sigma^{2}$
    \end{enumerate}
    
 \caption{Recover Distance}
 \label{alg:algorithm2}
\end{algorithm}

\subsection{Guarantees}
\begin{claim}\label{claim6}
Let $S \subset \mathbb{R}^{d}$ with $|S| = n$. Then, given any two entries in this set $x_i, x_j \in S$, let $y_i = x_iP + \Delta_i$ and $y_j = x_jP + \Delta_j$, where $P$ and $\Delta_i, \Delta_j$ are the projection matrix and the noise vectors respectively. Let, $\mathcal{D}(\cdot, \cdot)$ be defined as in Algorithm 2. Then, $\mathcal{D}(y_i, y_j)$ is an unbiased estimator of $||x_i - x_j||_{2}^{2}$. 
\end{claim}

\begin{proof}
Let $\Delta = \Delta_i - \Delta_j$, and let $\sigma^2$ be the variance of the entries of the projection matrices $\Delta_i$ and $\Delta_j$. Then we have: \\

\begin{align}
&\mathbb{E}\left[\mathcal{D}(y_i,y_j)\right] = \mathbb{E}\left[ ||x_iP + \Delta_i - x_jP - \Delta_j||_2^{2} - 2k\sigma^{2} \right] \nonumber \\
&= \mathbb{E}\left[ ||(x_i- x_j)P + \Delta||_2^{2} - 2k\sigma^{2} \right] \nonumber\\
&= \mathbb{E}\left[ ||(x_i- x_j)P ||_2^{2} + ||\Delta ||_2^{2} + 2\langle (x_i - x_j)P, \Delta  \rangle - 2k\sigma^{2} \right] \nonumber\\
\begin{split}
&= \mathbb{E} [ ||(x_i- x_j)P ||_2^{2}] + \mathbb{E}[||\Delta ||_2^{2}] + \\ 
&\enspace \enspace \enspace \enspace 2\mathbb{E}[\langle (x_i - x_j)P, \Delta  \rangle] - 2k\sigma^{2}
\end{split} \nonumber \\
\begin{split}
&= \mathbb{E} [ ||(x_i- x_j)P ||_2^{2}] + \mathbb{E}[\sum_{t=1}^{k}(\Delta_{i_{t}} - \Delta_{j_{t}})^{2}] + \\ 
& \enspace \enspace \enspace \enspace 2\mathbb{E}[\langle (x_i - x_j)P, \Delta  \rangle] - 2k\sigma^{2} 
\end{split} \nonumber \\
\begin{split}
&= \mathbb{E} [ ||(x_i- x_j)P ||_2^{2}] + \mathbb{E}[\sum_{t=1}^{k}(\Delta_{i_{t}})^{2} + (\Delta_{j_{t}})^{2} - 2\Delta_{i_{t}}\Delta_{j_{t}}] + \\
& \enspace \enspace \enspace \enspace  2\mathbb{E}[\langle (x_i - x_j)P, \Delta  \rangle] - 2k\sigma^{2} 
\end{split} \nonumber \\
&= \mathbb{E} [ ||(x_i- x_j)P ||_2^{2}] + \sum_{t=1}^{k}4b^{2} + 2\mathbb{E}[\langle (x_i - x_j)P, \Delta  \rangle] - 2k\sigma^{2} \nonumber\\
&= ||x_i- x_j ||_2^{2} + 4kb^{2} + 0 - 2k\sigma^{2} \nonumber\\
&= ||x_i- x_j ||_2^{2} + 4kb^{2} + 0 - 2k(2b^2) \nonumber\\
&= ||x_i- x_j ||_2^{2} \nonumber
\end{align}

Note that by the Johnson-Lindenstrauss lemma, we have $ \mathbb{E} [ ||(x_i- x_j)P ||_2^{2}] = ||x_i- x_j ||_2^{2}$. We now show that $2\mathbb{E}[\langle (x_i - x_j)P, \Delta  \rangle] = 0$,  \\

Letting $\mathbf{a} = (x_i - x_j)$, we have \\

\begin{align}
2\mathbb{E}[\langle (x_i - x_j)P, \Delta  \rangle] &= 2\mathbb{E}[\langle \mathbf{a}P, \Delta  \rangle] \nonumber\\
&= 2\sum_{t=1}^{k}\mathbb{E}[(\mathbf{a}P)_t] \mathbb{E}[\Delta_t] \nonumber\\
&= 2\sum_{t=1}^{k}\mathbb{E}[(\sum_{m=1}^{d}a_mP_{mt})] \mathbb{E}[\Delta_t] \nonumber\\
&= 2\sum_{t=1}^{k}(\sum_{m=1}^{d}a_m\mathbb{E}[P_{mt}]) \mathbb{E}[\Delta_t] \nonumber\\
&=0\nonumber
\end{align}

where we have used the independence of $P$ and $\Delta$. 

\end{proof}

\begin{claim}\label{claim7}
Let $S \subset \mathbb{R}^{d}$ with $|S| = n$. Then, given any two entries in this set $x_i, x_j \in S$, let $y_i = x_iP + \Delta_i$ and $y_j = x_jP + \Delta_j$, where $P$ and $\Delta_i, \Delta_j$ are the projection matrix and the noise vectors respectively. Then the variance of $\mathcal{D}(y_i,y_j) = \frac{2}{k}||x_i - x_j||_2^{4} + 2k(7 \sigma^{4} - \sigma^{2}) + 4 \sigma^{2}||x_i - x_j||_2^{2} $, where $\sigma^2$ is the variance of the entries of $\Delta_i$ and $\Delta_j$. 
\end{claim}

\begin{proof}
Let, \\
\begin{align}
Z_1 &= ||(x_i- x_j)P ||_2^{2}\nonumber \\
Z_2 &= ||\Delta ||_2^{2} \nonumber\\
Z_3 &= 2\langle (x_i - x_j)P, \Delta  \rangle \nonumber
\end{align}

Then, Var$(||x_iP + \Delta_i - x_jP - \Delta_j||_2^{2} - 2k\sigma^{2}) = \text{Var} (Z_1 + Z_2 + Z_3) - 2k \sigma^{2}$. Then, \\

\begin{align}
\begin{split}
\text{Var}(Z_1 + Z_2 + Z_3) &= \mathbb{E}[(Z_1 + Z_2 + Z_3)^{2}] \\
&- (\mathbb{E}[Z_1 + Z_2 + Z_3])^{2} 
\end{split}\nonumber\\
\begin{split}
&= \mathbb{E}[Z_1^{2}] - \mathbb{E}[Z_1]^{2} + \mathbb{E}[Z_2^{2}] -  \\
&\mathbb{E}[Z_2]^{2} + \mathbb{E}[Z_3^{2}] - \mathbb{E}[Z_3]^{2} + 
\end{split} \nonumber\\
\begin{split}
 &2\mathbb{E}[Z_1 Z_2] - 2\mathbb{E}[Z_1]\mathbb{E}[Z_2] +\\
 &2\mathbb{E}[Z_2 Z_3] - 2\mathbb{E}[Z_2]\mathbb{E}[Z_3]
\end{split} \nonumber \\
 &+  2\mathbb{E}[Z_1 Z_3] - 2\mathbb{E}[Z_1]\mathbb{E}[Z_3] \nonumber\\
 \begin{split}
&= \frac{2}{k}||x_i - x_j||_2^{4} + 14k\sigma^{4} + \\
&4\sigma^{2}||x_i - x_j||_2^{2} 
\end{split}\nonumber\\
\end{align}

where we have used the following : \\
\begin{align}
\mathbb{E}[Z_{1}^{2}] - \mathbb{E}[Z_1]^{2} &= \frac{2}{k}||x_i - x_j||_2^{4} \nonumber\\
\mathbb{E}[Z_2^{2}] - \mathbb{E}[Z_2]^{2} &= 14k\sigma^{4}\nonumber \\
\mathbb{E}[Z_3^{2}] - \mathbb{E}[Z_3]^{2} &= 4\sigma^{2}||x_i - x_j||_2^{2} \nonumber
\end{align}

Using independence of $Z_1$ and $Z_2$ we have $2\mathbb{E}[Z_1 Z_2] = 2\mathbb{E}[Z_1]\mathbb{E}[Z_2]$. For the rest of the variables we have $2\mathbb{E}[Z_2 Z_3] = 2\mathbb{E}[Z_2]\mathbb{E}[Z_3] =  2\mathbb{E}[Z_1 Z_3] = 2\mathbb{E}[Z_1]\mathbb{E}[Z_3] = 0$.  Using these, the required result follows. 

\end{proof}

We can consider the probability of the distance recovery algorithm exceeding a fixed error $\lambda$. More specifically, letting $x_i, x_j$ be two points in the original space and letting $y_i, y_j$ be the points after the action of the mechanism, we want to know how this value is bounded : $\mathbf{Pr}\left[ \left| \mathcal{D}(y_i,y_j) - ||x_i - x_j||_2^{2} \right| > \lambda \right]$. Using the Chebychev inequality, we get:\\
\begin{align}
\mathbf{Pr}\left[ \left| \mathcal{D}(y_i,y_j) - \mathbb{E}[\mathcal{D}(y_i,y_j)] \right| > \lambda \right] \leq \frac{Var(\mathcal{D}(y_i,y_j))}{\lambda^2} \nonumber \\ 
\therefore \mathbf{Pr}\left[ \left| \mathcal{D}(y_i,y_j) - ||x_i - x_j||_2^{2} \right| > \lambda \right] \leq \frac{Var(\mathcal{D}(y_i,y_j))}{\lambda^2} \nonumber \nonumber
\end{align}

Therefore, we see that the probability the distance recovery algorithm exceeds a fixed error is proportional to the distance between the original points and the dimensionality of the subspace we project the data into.

\begin{figure}[htb]
\begin{tabular}{ c }
{\includegraphics[width=3in, height=1.8in]{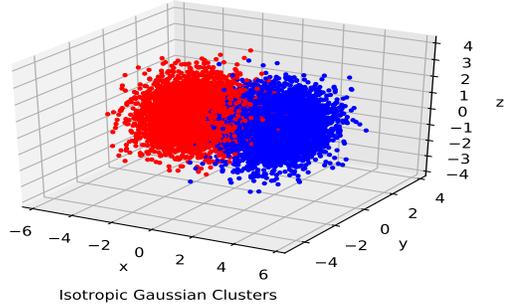}} 
\end{tabular}
\caption{Scatter plot of the original 3D dataset with colors separating the two clusters}
\label{fig:3d-vis}
\end{figure}

\begin{figure}[htb]
\begin{tabular}{ c }
{\includegraphics[width=3in, height=1.8in]{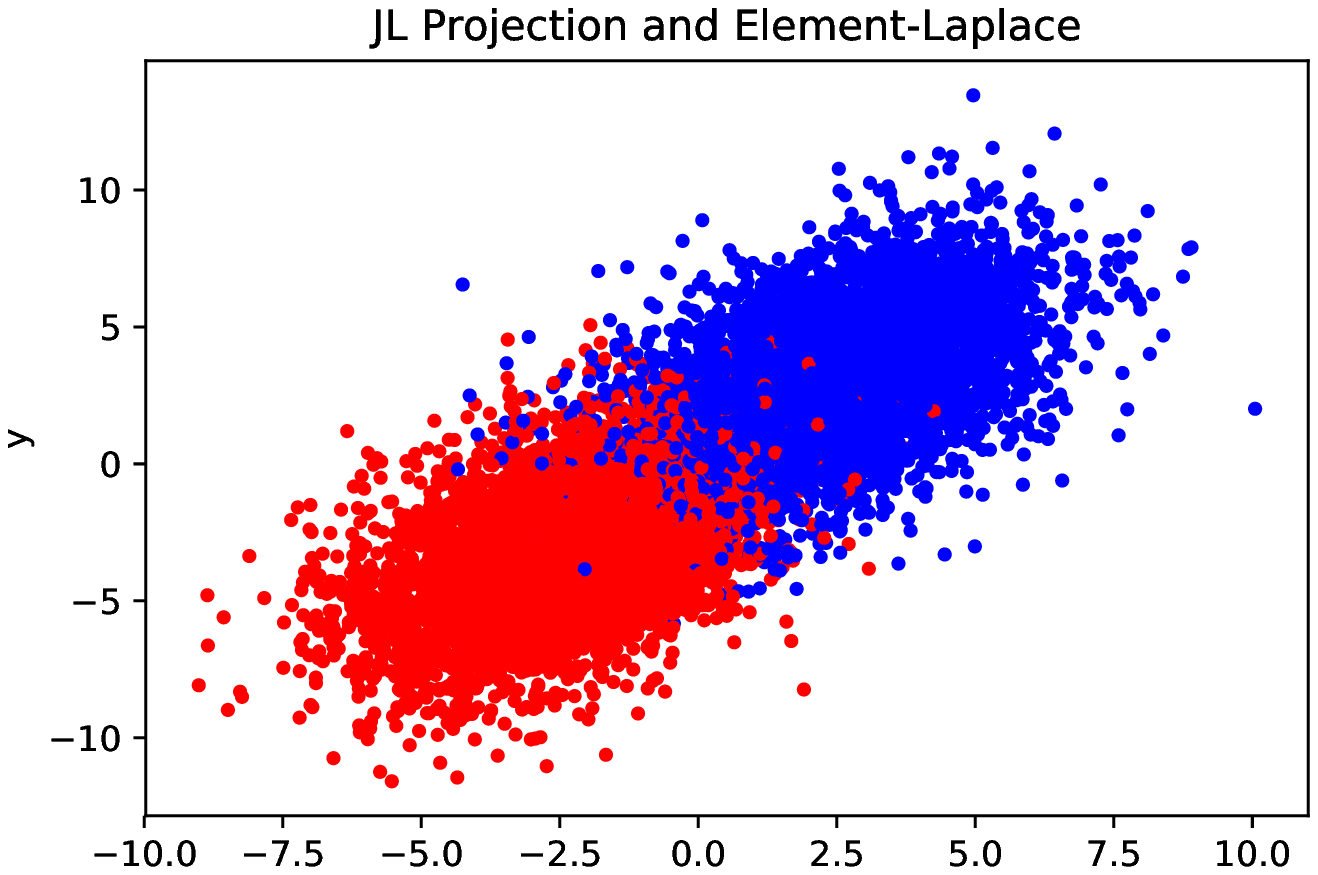}} \\
{\includegraphics[width=3in, height=1.8in]{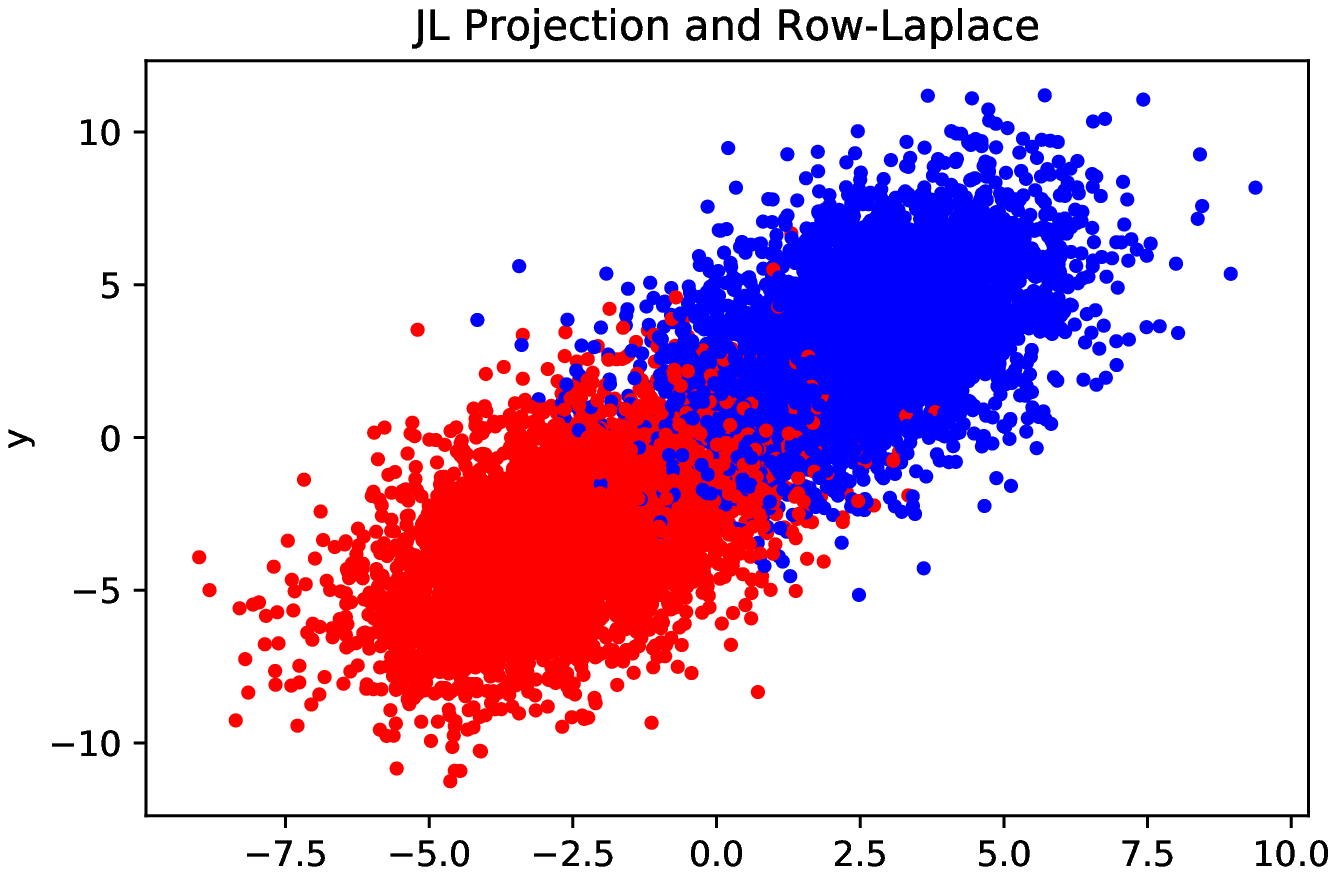}} \\
\end{tabular}
\caption{Scatter plots of the 2D dataset for the Element-wise privacy mechanism and the Row-wise privacy mechanism with colors separating the two clusters. Each privacy mechanism turns the spherical data into ellipses that are stretched along the directions of the noise while still maintaining separation.}
\label{fig:2d-proj}
\end{figure}

\section{Experiments}
\label{experiments}

All of our experiments are based on a synthetic dataset comprising of two clusters generated through the procedure developed by Guyon (2003), which is commonly referred to as the Madelon dataset. In this data generation procedure, the data points for each cluster are sampled from an isotropic Gaussian distribution. We fix the Euclidean distance between the cluster centers to $4$ and use an $\epsilon$ of 4 for all our experiments. We also assume that for two neighboring databases (element-wise or row-wise), the norm of their difference is bounded by 1. That is, for two neighboring databases $X,X' \in \mathbb{R}^{n \times d}$, we have $||X - X'||_1 \leq 1$.  We use the open source implementation of this data generation procedure provided in Scikit as \emph{sklearn.datasets.make\_blob}  (Pedregosa et al. (2011)).

Since this dataset is decoupled from any specific problem domain and is a two class clustering problem, it allows us to demonstrate the utility of our mechanism with maximum generality. In order to better understand the generated data in higher dimensions, we first generate data in $\R^3$ and provide its visualization in Figure \ref{fig:3d-vis}. Next, we illustrate the effect of the element-wise and  row-wise privacy mechanisms defined in Algorithm \ref{alg:algorithm1}, by visualizing the data using $k=2$ (i.e. projecting it into $\R^2$ and making it private) in Figure \ref{fig:2d-proj}.   

Using the same dataset defined above, we verify the utility guarantees of our distance recovery algorithm (Algorithm \ref{alg:algorithm2}).  In order to do so, we first sample 1000 pairs of points from the original dataset and calculate the squared Euclidean distance between each pair. This gives us a total of 1000 distances. We then run each pair through our privacy mechanism 1000 times using a new projection and perturbation vector each time (giving us 1 million private pairs). We then use our distance recovery algorithm ( Algorithm \ref{alg:algorithm2} ) to calculate the squared Euclidean distance between each private pair of points, giving us 1 million distances for the private data points (1000 distances for each private pair). Next, we plot the distribution of differences in the squared Euclidean distance between the original pair and the private pairs in Figure \ref{fig:distance_recov}.  One can see that the distance recovery algorithm does indeed recover the squared Euclidean distances in expectation. For one run of the experiment, the mean of the differences for element-wise private pairs and row-wise private pairs from the original pairs was found to be $0.006$ and $-0.011$ respectively.

\begin{figure}[htb]
\begin{tabular}{ c }
{\includegraphics[width=3in,height=1.8in]{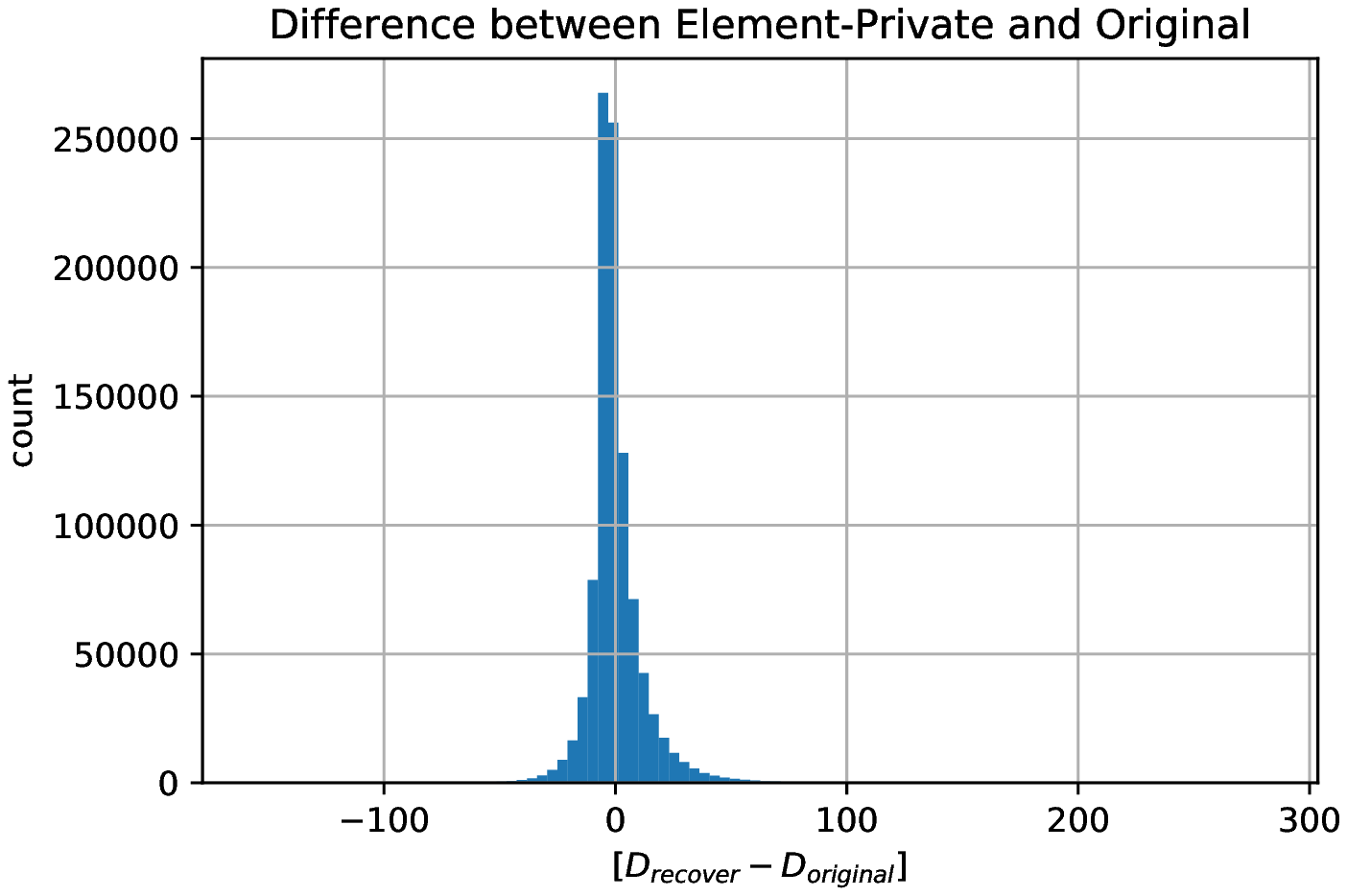}} \\ 
{\includegraphics[width=3in,height=1.8in]{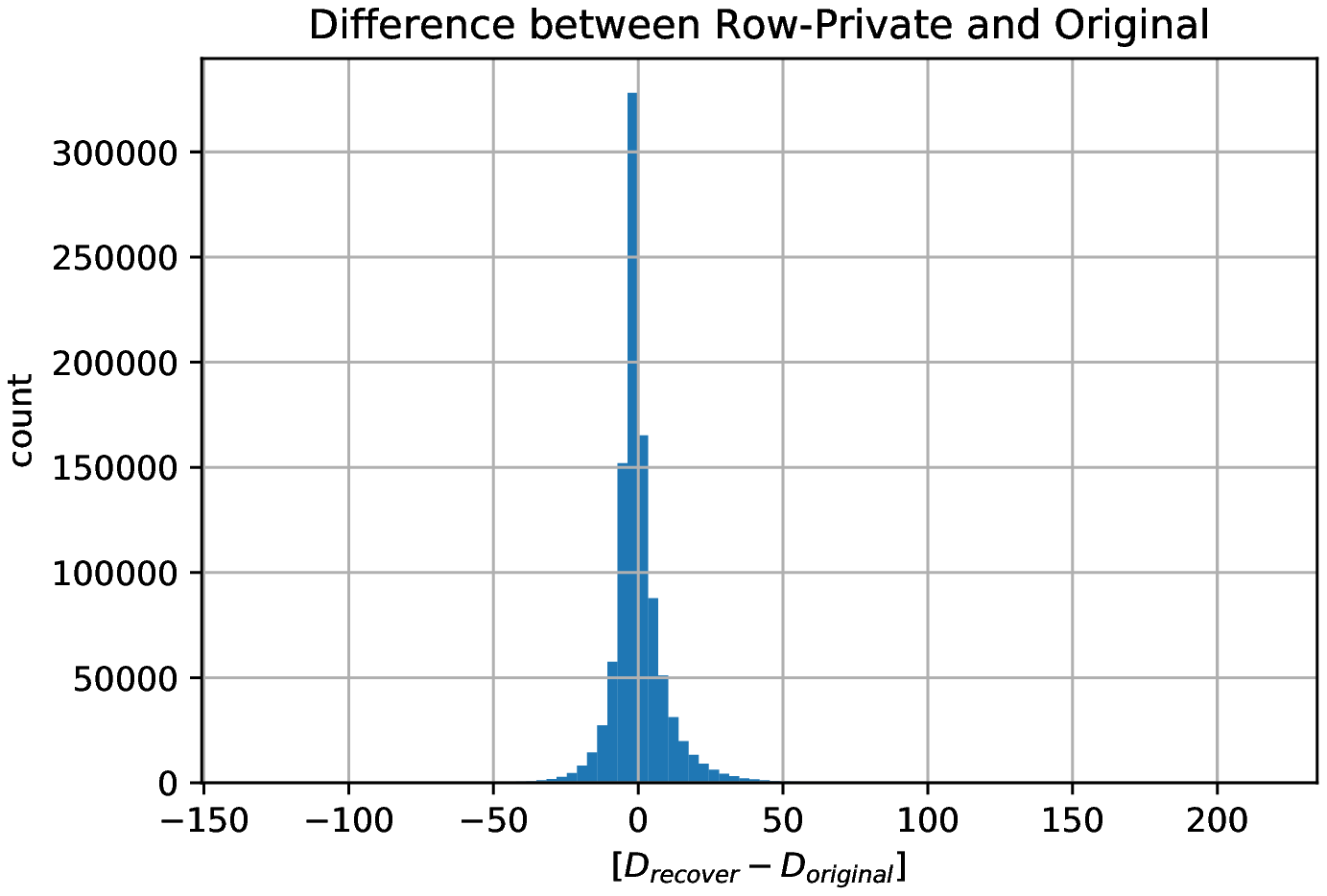}} \\ 
\end{tabular}
\caption{Distribution of the difference between the original squared Euclidean distances and the squared Euclidean distances recovered by our distance recovery algorithm.   }
\label{fig:distance_recov}
\end{figure}

In order to test the utility of our method on the task of clustering, we compare the performance of the k-means clustering algorithm on the original data, element-wise private data, and the row-wise private data. We ran this comparison for a number of datasets in which we varied both the original and projected dimensions. The results of this experiments are provided in Table \ref{tab:Clustering Accuracy}. We note that the mechanism provides good utility for smaller $k$ but the utility deteriorates as we increase the dimensionality of the projected subspace, a result that is expected due to the reliance of Laplacian noise on the projection subspace parameter $k$ as shown in Theorems \ref{theorem3}, and \ref{theorem5}.

We examine the relationship in more detail by plotting the relationship between $k$ and the standard deviation of the  data in Figure \ref{fig:std_inc}. The formula used for this is $\sqrt{1 + 2b^2}$, where 1 is the variance of the original data and $2b^2$ is the variance of the Laplacian noise. It can be noted that the standard deviation increases with $k$ hence negatively affecting the amount of utility provided by the mechanism. We also note a difference in the standard deviation between element-wise and row-wise privacy mechanisms - row-wise privacy comes at a higher cost utility cost than element-wise privacy.

\begin{table*}[t]
\begin{center}
\begin{tabular}{|c|c|c|c|c|c|}
\hline
{\bf Privacy Mechanism} & {\bf $d$=3, $k$=2} & {\bf $d$=10, $k$=3} & {\bf $d$=50, $k$=10} & {\bf $d$=100, $k$=20}  \\
\hline
None & 0.9783 & 0.9772&0.9771 &0.9797 \\
\hline
Element-Wise & 0.9441 & 0.9082 &0.6954 &0.6927\\
\hline
Row-Wise & 0.9477 & 0.909 &0.6796  &0.6668 \\
\hline
\end{tabular}
\end{center}
\caption{Comparison of performance of k-means clustering between the non-private, element-wise private, and row-wise private data. Here $d$ is the original dimension of the data and $k$ is the projected dimension.}
\label{tab:Clustering Accuracy}
\end{table*}

\begin{figure}[htb]
\begin{tabular}{ c }
{\includegraphics[width=3in,height=1.8in]{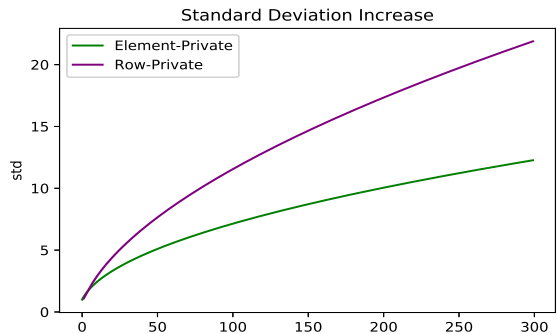}} \\ 
\end{tabular}
\caption{Increase in standard deviation with $k$ }
\label{fig:std_inc}
\end{figure}

Our experiments validate the ability of the distance recovery algorithm (Algorithm \ref{alg:algorithm2}) to recover the squared Euclidean distances between original points and also show that our privacy mechanisms are able to maintain utility in the general task of clustering. We find that the utility of our mechanism deteriorates with an increase in the dimensionality of the projection subspace which agrees with  Theorems \ref{theorem3}, and \ref{theorem5}.

\section{Related Work}
\label{related_work}
Differential privacy is a framework proposed by \cite{dwork2006calibrating} that captures precisely how much additional information of an individual is leaked by participating in a database, that would not have been leaked otherwise. There has been extensive research in proposing mechanisms that guarantee differential privacy in the non-interactive setting \cite{alda2017bernstein, balog2017differentially, mcsherry2007mechanism, dwork2014algorithmic}.  

\cite{kenthapadi2012privacy} developed a randomization mechanism that utilized the JL transform and the Gaussian mechanism \cite{dwork2014algorithmic} to provide non-interactive differential privacy with respect to attribute changes. They showed that their mechanism preserved utility by preserving distances in expectation. However, a shortcoming of this approach was that the privacy guarantees were only provided with respect to attribute changes, and not row level changes, which is a more realistic requirement in practice. Despite that shortcoming, the mechanism was powerful from a privacy perspective, as it had been shown by \cite{liu2006random}  that random projection-based multiplicative perturbation techniques make it impossible to find the exact values of the original data in addition to simply hiding the dimensionality  of the data. Further, they showed that if even if the projection matrix is released, the adversary still cannot find the exact value of any elements from the original data. 

\cite{blocki2012johnson} showed that the JL transform itself preserved differential privacy and provided utility guarantees in the strict case when only the covariance matrix is released. However, in order to provide privacy guarantees, the data matrix was required to be full rank with eigenvalues above some threshold. Since this is not always feasible in practice, they provided a work around which perturbed all the singular values of the data matrix. In practice, this magnitude of this perturbation can be orders of magnitude larger than the attribute values, hence causing general machine learning algorithms to have extremely poor performance. Along similar lines of using multiplicative random projections to preserve privacy for special problems is the work of \cite{zhou2009differential} who showed that multiplicative random projection methods preserved utility in the case of doing PCA.  

Releasing differentially private data raises some fundamental questions about the ability of machine learning algorithms to extract utility from the private data. \cite{kasiviswanathan2011can} showed that in the PAC learning model with a discrete domain, any finite hypothesis class that is PAC learning is also privately PAC learnable. These results were extended to half space queries by \cite{blum2013learning} and the sample complexities of proper and improper learners were analyzed by \cite{beimel2010bounds}. However, \cite{chaudhuri2011sample} showed that there exist simple hypothesis classes over continuous domains that have a small VC dimension and for whom it is impossible learn privately with a finite sample size. \cite{friedman2010data} analyzed the trade-off between privacy, sample complexity, and utility in practice for the case of decision trees. 

Another line of research focused on releasing differentially private models with respect to the data \cite{chaudhuri2009privacy, evfimievski2004privacy, sheffet2015differentially, zhu2017differentially}. \cite{chaudhuri2011differentially} developed a mechanism for private empirical risk minimization that provided private approximates to classifiers and along similar lines \cite{bassily2014differentially} analyzed error bounds on such classifiers. Releasing private models also raised questions between the trade-off of privacy and algorithmic complexity, which was analyzed by \cite{friedman2010data} in practice for the case of decision trees. 

\section{Conclusions and Future Work}
We developed a privacy mechanism that makes it impossible to reconstruct the original data values while also providing utility for general machine learning tasks. We proved privacy guarantees under element and row wise changes, and also proved utility guarantees by proposing an algorithm that maintains pairwise distances between private data points in expectation. We chose the most general task of clustering in order to show the effectiveness of our methods experimentally and validated that it does in fact maintain utility. Noting that the utility of our mechanism deteriorates with an increase in the dimensionality of the projection subspace, we leave open the question of finding a private mechanism that makes reconstruction impossible and provides utility that does not deteriorate with the dimensionality of the problem. 
 
\bibliography{example_paper}
\bibliographystyle{icml2018}

\end{document}